\documentclass{article}
\usepackage{float}
\usepackage{authblk}

\title{Multi-server Blind Quantum Computation Protocol with Limited Classical Communication among Servers}
\author[1]{Yuichi Sano\thanks{\texttt{sano.yuichi.77v@st.kyoto-u.ac.jp}}} 
\affil[1]{Department of Nuclear Engineering, Kyoto University, Nishikyo-ku, Kyoto 615-8540, Japan} 

\date{\today}

\usepackage{graphicx}
\usepackage{amsmath,amsthm,amssymb}
\usepackage{amsfonts}
\usepackage{siunitx}

\usepackage{qcircuit}
\usepackage{braket}
\usepackage{cite}
\usepackage{comment}

\theoremstyle{definition}
\newtheorem{thm}{Theorem}
\newtheorem{defn}[thm]{Definition}
\newtheorem{lem}[thm]{Lemma}

\begin{document}

\maketitle

\begin{abstract}
A user who does not have a quantum computer but wants to perform quantum computations may delegate his computation to a quantum cloud server.
For users securely using the service, it must be assured that no malicious server can access any vital information about the computation.
The blind protocol was proposed as a mechanism for users to secure their information from unauthorized actions of the server.
Among the blind protocols proposed thus far, a protocol with two servers sharing entanglement does not require any quantum resource from the user but does not allow the servers to interact even after the computation.
We propose a protocol in this paper that extends this two-server protocol to multiple servers and is secure even if some servers communicate with each other after the computation.
Dummy gates and a circuit modeled after brickwork states play a crucial role in the new protocol.
We also show that it is possible to make estimates for risks that could not be estimated in a previous study.
\end{abstract}

\section{Introduction}
Quantum computers have been actively studied with the expectation of providing higher computational capacity than classical computers.
For example, Shor's algorithm uses the quantum Fourier transform to solve prime factorization and discrete logarithm problems exponentially faster than existing conventional algorithms.
Grover's algorithm is the fastest searching algorithm for an unordered database, at a speed that is thought to be impossible to achieve with classical computation\cite{Grover}.
In addition to specific algorithms, it is known that classical computers cannot sample as fast as quantum computers in the sampling problem\cite{sampling,Fine-Grained}.
While quantum computers have such superiority, they will be enormously expensive compared with classical ones even if they will become available in the future as they will need some fine-tuned microscopic devices to use quantum effects..

When quantum computers are used as cloud servers, user security is a concern.
The user needs to send information about his calculations to the server to delegate the calculations.
If the server is malicious, it may illegally obtain the user's information.
Therefore, the user must do a blind quantum computation\cite{Childs,BFK,MF protocol,Morimae hayashi,Reichardt,McKague,Sano}.
A blind quantum computation protocol securely encrypts inputs, outputs, and calculation process of the calculations delegated by the user to the server.
Blind quantum computation is expected to be an advantage of the new quantum computation because it is more powerful than fully homomorphic encryption, which is its classical analog\cite{FHE}.

At the moment, a blind quantum computation cannot be performed unconditionally; therefore, servers and users must be subjected to constraints.
A commonly studied blind protocol is with the user and only a single server\cite{Childs,BFK,MF protocol,Morimae hayashi,Sano}.
With the protocol of using the single server, the user does not have to impose any restrictions on the server, but the user's abilities must be higher than the abilities of classical communication and classical computation.
However, in the protocol of using multiple servers, the user only needs to have the abilities of classical computation and classical communication to delegate his calculation\cite{Reichardt,McKague}.
Classical computation and communication abilities are now widely available, and they are thus considered to be essential abilities for users.
Therefore, these protocols using multiple servers are very convenient for users.
However, very strong restrictions are imposed on the server: servers cannot do classical communication with each other.
In reality, we cannot assume that servers do not have the ability to perform classical communication with each other, so we can assume that servers do not perform classical communication with each other according to the contract with the user.
There is no problem if the server honors the contract; however, we need a blind protocol in case the server is malicious in the first place.
In addition, servers that initially honor their contracts with users may suddenly break those contracts.
The current protocol does not allow users to estimate the extent of those risks.
The purpose of our study is to relax the restrictions imposed on servers and to provide a way for users to estimate security against malicious servers.

In this study, we propose a blind protocol using multiple servers, in which some servers can classically communicate after the computation.
First, we extend the protocol from two servers case discussed in previous studies.
Next, we assume the situation where some servers can do classical communication after a calculation.
We then propose a method of encrypting the circuits used in the calculation so that the user can delegate the calculation by the blind computation even if some servers are in classical communication after the calculation.
Moreover, we show that if the user delegates his calculation to sufficiently many servers, the risk of their knowledge about the user's computation can be estimated.

\section{Preliminaries}
In this section, we describe a blind protocol and explain the blind protocol with two servers, which is the basis of the protocol proposed in our study.

\subsection{Blind Quantum Computation Protocol}
In this subsection, we first describe a blind protocol.
The blind protocol, first proposed by Childs\cite{Childs}, is a security protocol that hides a user's input/output and calculation process from a server.
To be more specific, when a user uses blind protocol to delegate his computation, the server has no way of knowing whether the multiple delegated calculations are the same or even different.
A classical analog is a fully homomorphic encryption, where the input and output are encrypted while the calculation is performed\cite{FHE}.
The fully homomorphic encryption server delegated the calculation has no knowledge of the input/output, but it knows about the operation being performed.
Therefore, the blind protocol is superior to the fully homomorphic encryption in terms of security because it can also hide the calculation process.
There are currently no known blind protocols in classical computation using classical computer servers.
For this reason, blind protocols are a unique advantage of quantum computation.

Broadbent, Fitzsimons, and Kashefi give the following definition of the blind protocol\cite{BFK}.
\begin{defn}[Blindness{\cite[Definition 2]{BFK}}]
\label{defn:Blindness}
Let P be a quantum delegated computation on input X and let L(X) be any function
of the input. We say that a quantum delegated computation protocol is blind while leaking at most
L(X) if, on user's input X, for any fixed Y = L(X), the following two conditions when given Y:
\begin{itemize}
 \item[1.] The distribution of the classical information obtained by server in P is independent of X.
 \item[2.]Given the distribution of classical information described in 1, the state of the quantum system
obtained by server in P is fixed and independent of X.
\end{itemize}
\end{defn}

The blind protocol is a computational protocol that satisfies definition \ref{defn:Blindness}.
This definition allows the server to know calculation-independent information such as the size of gates and information about protocol instructions but does not allow it to obtain calculation-dependent information that can distinguish between any calculations.

\subsection{Two Server Blind Protocol}
The first proposed blind protocol is performed by a user with quantum memory and a single quantum server\cite{Childs}.
Since then, blind protocols have been proposed, but blind protocols performed between a single server and a user require the user to have quantum abilities\cite{Childs,BFK,MF protocol,Morimae hayashi,Sano}.
Currently, there is no known blind protocol that can be performed by the user having only abilities of classical computation and classical communication with a single server\cite{Fitzsimons review}.
It is also not known about the possible or impossible existence of the blind protocol with a single server with the user who has only classical computation and classical communication abilities, but some negative results have been obtained\cite{morimae hitei,aronson hitei}.

In contrast, blind protocols with two servers are available for users who can only perform classical computation and classical communication\cite{Reichardt,McKague}.
These blind protocols are executed by a user making individual classical communication with multiple servers that share entanglement states.
However, classical and quantum communication among servers is prohibited during and after the computation.
The ability to perform classical computation and classical communication is so common that, for example, a single desktop computer is sufficient, which is very convenient for users.

In this paper, we propose a protocol based on the protocol\cite{Reichardt} that allows blind computation even after some servers have performed classical computation.
For this purpose, we first describe the protocol in \cite{Reichardt}.
\begin{thm}[${\rm MIP^*=QMIP}$\cite{Reichardt}]
\label{QMIP}
\begin{equation*}
    {\rm MIP^*=QMIP}.
\end{equation*}
\end{thm}
${\rm MIP^*}$ is a set of problems that can be verified via classical communication between a user with the ability of classical computation and servers who can unbounded computation and share entanglement among servers.
QMIP is a set of problems that can be verified via quantum communication between a user who has the ability of quantum computation and servers who have the ability of infinite computation.
With theorem\ref{QMIP}, the following lemma is also known for the case where the server(verifier) is limited.
\begin{lem}[Two Server Blind Quantum Computation Protocol\cite{Reichardt}]
\label{QMIP(0)}
\begin{equation*}
    {\rm MIP^*[2\ servers]\geq QMIP[0\ server]=BQP}.
\end{equation*}
\end{lem}
Lemma \ref{QMIP(0)} shows that a user can achieve the same results as a quantum computer by using classical communication with two servers that share entanglements.
This means that a user can do quantum computation without sending any information to the servers.
We do not go into detail about the protocol, but it's based on the user monitoring the server's activity utilizing the two servers' entanglement.
Servers do not communicate with each other by assumption because this protocol is built on the interactive proof system of computational complexity theory.
If users try to implement this protocol in the real world, they will have to sign a contract prohibiting the server's classical communication.
The limitations placed on this server are extremely stringent.
In the following section, we propose a partial relaxation of this restriction.

\section{Multi-server Blind Quantum Computation Protocol with Limited Classical Communication among Servers}
In this section, we propose a protocol that allows the user's computation to remain blindness even if some servers perform classical communication with other servers after the computation that the user delegates for servers.
To do that, we first show that it is possible to perform the protocol on several servers, as extended from the two servers' protocol of the previous study\cite{Reichardt}.
Next, we show how to encrypt the circuit so that even if the server gets some information about the circuit, it cannot know anything about the calculation that the user delegated to it.

\subsection{Extension to Multiple Servers}
In principle, the protocol from the preceding study is performed by two servers and a user using classical communication.
When the number of servers is increased from two to many, each server has less knowledge about the calculation.
One trivial way to increase the number of servers is to add virtual servers that do not participate in the calculation, but we will not consider this.
We suggest a non-trivial method for increasing the number of servers by internally separating each of the protocol's two servers from the prior study.
\begin{thm}
\label{multi}
Even though the internal roles of the two servers in the lemma \ref{QMIP(0)} protocol are split into multiple servers each, the protocol is still a blind protocol.
\end{thm}
\begin{proof}
We refer to the two servers used by lemma \ref{QMIP(0)} as server A and server B, respectively.
We will split these servers into several groups.
The set of servers that split server A is $\{A_a\}_a$, and the set of servers that split server B is $\{B_b\}_b$, where $a$ and $b$ are the number of server A and server B splits.
Let $\{m^A_a\}_a$ and $\{m^B_b\}_b$ be the set of messages that a user sends to each server.

We use proof by contradiction.
We assume that the server can break blindness and obtain information about the user's calculation using this server-splitting protocol.
By assumption, the server is getting information from the messages $\{m^A_a\}_a$ and $\{m^B_b\}_b$ with the user that would break blindness.
Let $m_A$ and $m_B$ be the sets of messages received by the original server A and server B.
The server split is just a split of the internal workings of the original server, so $\{m^A_a\}_a$ can be created from $m_A$.
The $\{m^B_b\}_b$ can be created in the same way.
Thus, the original server A and server B can easily simulate the server's behavior after the split.
Therefore, the original server A and server B can also obtain information that breaks the blindness.
However, this contradicts the fact that the protocol consisting of server A and server B is a blind protocol.
Hence, the assumption is wrong, i.e., the protocol will remain a blind protocol even if the server is split such that the internal roles of server A and server B are split.
\end{proof}
It turns out that the lemma \ref{QMIP(0)} can be solved using two or more multiple servers. 
As the number of servers used for computation grows, the circuits that use the computation become more fragmented, and each server knows less about the user's calculation.
In reality, the amount of information each server knows is irrelevant if the servers do not use classical communication.
However, each server needs to have a small amount of information if the servers perform classical communication with each other after the computation.
If some servers are allowed to communicate with other servers after the calculation, two servers can quickly learn the entire circuit if the protocol of the previous study.
But if the number of servers participating in the calculation becomes huge, it becomes difficult to know everything completely.
Of course, part of the circuit depends on the calculation, so the protocol is no longer a blind protocol under such an assumption.
In the next subsection, we will encrypt the circuit so that servers can get some information about the circuit, but not the information that depends on the calculation of the user.

\subsection{Main Protocol}
We propose a blind protocol in this subsection, even if some servers communicate after the computation and transfer information about the computation to one another.
The protocol used in the prior study, the protocol comprises two servers: one server that runs quantum gates on qubits and another that receives those qubits once and returns them to the first server.
Thus, the entire quantum circuit for a calculation is realized by one of the two servers.
From the information in the quantum circuit, the server knows the input and output of the calculation and the calculation process.
Therefore, when the server is split up, if the server on the side running the quantum circuit shares information, information about the user's calculations will be leaked.
Hence, we propose a method to encrypt the circuit so that even if the server knows some information about the circuit, the information does not depend on the user's calculation.

In the previous study, $\{CNOT,G\}$ is used as the universal gate set\cite{universal}.
$G$ gate is a gate that looks like
\begin{equation}
 G = R_y(\frac{-\pi}{8}).
\end{equation}
We adopt $\{H,T,CZ\}$ as a universal gate set for simplicity\cite{Nielsen-Chuang}.
These universal gate sets can approximate each other with polynomials, so the difference is not essentially significant\cite{Nielsen-Chuang}.

First, we define {\sl a circuit like brickwork states}\cite{BFK} to perform calculations with the same structure.
\begin{defn}[Circuit Like Brickwork States]
A circuit like brickwork states consists of fixed number $n$ of gates vertically and fixed number $p(n)$, is a polynomial of $n$, of gates horizontally.
The circuit is composed as follows.
We can always make $n$ and $p(n)$ an even number by adding ancilla gates.
The circuit is created in the following steps.
\begin{description}
    \item[Step 1.] Each row starts with arbitrary unitary operator $V$ which consists  of $m$ gates.
    \item[Step 2.] Then, a unitary operator $U_1$ is put on the $l$-th row, where $l$ is $\{l=2k+1|k=0,1\ldots,\frac{n}{2}-1\}$, and a unitary operator $U_2$ is put on $l+1$ th row.
    \item[Step 3.] The user puts the $CZ$ gate between the $l$-th row and the $l+1$ th row.
    \item[Step 4.] A unitary operator $U_3$ is put on the $l$-th row, and a unitary operator $U_4$ is put on $l+1$ th row.
    \item[Step 5.] The user performs arbitrary unitary operator $V$ which consists of $m$ gates on each row.
    \item[Step 6.] The unitary operator $U_2$ is put on the $l$-th row, and the unitary operator $U_1$ consisting of four gates is put on $l-1$ th row.
    The user does not execute anything in the first and last row between steps 6 to 8.
    \item[Step 7.] The user puts the $CZ$ gate between the $l$-th row and the $l-1$ th row.
    \item[Step 8.] The unitary operator $U_4$ is put on the $l$-th row, and the unitary operator $U_3$ is put on $l-1$ th row.
    \item[Step 9.] Repeat steps 1 to 8 until the last column is reached.
\end{description}
The unitary operators $\{U_1,U_2,U_3,U_4\}$ that exist in front of each $CZ$ gate can be changed to the identity gate or the $CNOT$ gate by changing it as shown in Figs.\ref{fig:make identity}--\ref{fig:make CNOT}.
If $U_1=I,U_2=I,U_3=I,U_4=I$, the two $CZ$ gates, and the four unitary operators act as the identity gate.
If $U_1=R_z(\frac{\pi}{2}),U_2=R_x(\frac{\pi}{2}),U_3=I,U_4=R_x(\frac{-\pi}{2})$, the two $CZ$ gates and the four unitary operators act as the $CNOT$ gate.
Each unitary operator can be made with a combination of identity gate, $T$ gate, $T^{\dagger}$, and $H$ gate as shown below:
\begin{equation}
    I = H\cdot I \cdot H \cdot I,
\end{equation}
\begin{equation}
    R_z(\frac{\pi}{2}) = H\cdot I \cdot H \cdot T^2,
\end{equation}
\begin{equation}
    R_x(\frac{\pi}{2}) = H\cdot T^2 \cdot H \cdot I,
\end{equation}
\begin{equation}
    R_x(\frac{-\pi}{2}) = H\cdot T^{\dagger 2} \cdot H \cdot I.
\end{equation}

And we assume that the number $m$ of gates constituting an arbitrary unitary operator $V$ between $CZ$ gates is fixed.
The reason is that to be a blind protocol, $m$ needs to be fixed for any circuit, should not be determined for each circuit.
\end{defn}

\begin{figure}[t]
\centering
\includegraphics[width=8cm]{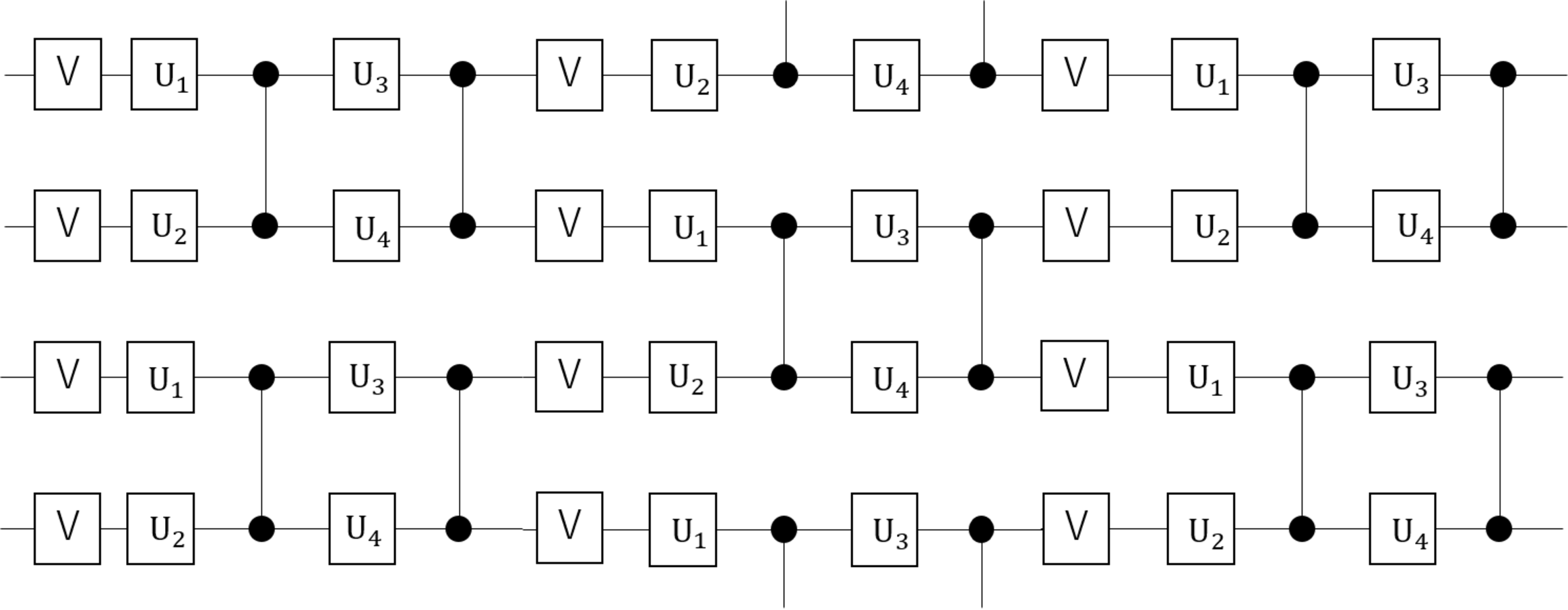}
\caption{Circuit like brickwork states.}
\label{fig:BFK}
\end{figure}

\begin{figure}[t]
\begin{equation*}
     \Qcircuit @C=1em @R=.7em {
   &  \qw & \gate{I}  & \ctrl{1} & \gate{I}  & \ctrl{1} & \qw \\
   &  \qw & \gate{I}  & \ctrl{0} & \gate{I}  & \ctrl{0}  & \qw 
    }
\end{equation*}
\caption{Combination of the $CZ$ gates and single qubit gates acting as the identity gate.}
\label{fig:make identity}
\end{figure}
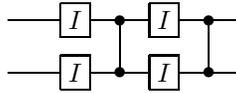

\begin{figure}[t]
\centering
\begin{equation*}
     \Qcircuit @C=1em @R=.7em {
   &  \qw & \gate{R_z(\frac{\pi}{2})}  & \ctrl{1} & \gate{I}  & \ctrl{1} & \qw \\
   &  \qw & \gate{R_x(\frac{\pi}{2})}  & \ctrl{0} & \gate{R_x(\frac{-\pi}{2})}  & \ctrl{0}  & \qw 
    }
\end{equation*}
\caption{Combination of the $CZ$ gates and single qubit gates acting as the $CNOT$ gate.}
\label{fig:make CNOT}
\end{figure}
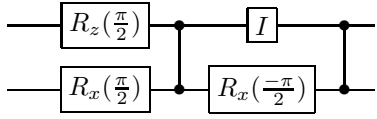

A circuit like brickwork states has a form like Fig. \ref{fig:BFK}.
Using this circuit like brickwork states, servers can know about the structure of the circuit but cannot get any information about the computation.
The user can run any circuit by setting $V$ and $U_i$ in Fig. \ref{fig:BFK} appropriately.

It is known that any single qubit gate can be made from a combination of $H$ gate and $T$ gate\cite{Nielsen-Chuang}.
Specifically, it approximates an arbitrary single qubit gate by rotating of two axes on the Bloch ball, $HTHT$ and $HT^{\dagger}HT^{\dagger}$.
Adding the identity gate $I = HIHI$ to these two sets allows the user to create any circuit by combining it with the circuit like brickwork states.

Next we discuss {\sl dummy gates} to eliminate information about the computation from these gate combinations.
\begin{defn}[Dummy gates]
Dummy gates consist of a combination of the following gates:
\begin{equation}
    D_1 =H T,
\end{equation}
\begin{equation}
    D_2 =H T^{\dagger},
\end{equation}
\begin{equation}
    D_3 =H I.
\end{equation}
For a constant $K$, a gate sequence that is a combination of all $K$ consecutive gates from $\{D_1, D_2, D_3\}$ is called dummy gates.
\end{defn}
As mentioned earlier, an arbitrary single qubit gate consists of $HTHT$, $HT^{\dagger}HT^{\dagger}$, and $HIHI$.
When a single qubit gate consists of $K/2$ combinations of these three, the dummy gates consisting of $K$ consecutive $\{D_1, D_2, D_3\}$ are added except for duplicate gates.
A server who knows nothing about the original single qubit gate will not distinguish between those dummy gates and the original.
By using those dummy gates, even if the server gets the information of $K/2$ consecutive gates from the user, the server cannot distinguish which gate is the original gate by running the dummy gates in parallel.
By definition, dummy gates do not allow the server to obtain any information about the calculation, even if the server knows about a gate combination of less than $K/2$.
Adding dummy gates to the calculation is efficient because it is an operation that adds at most $3^K$ gates when $K$ is a constant.

We show the procedure required to hide the output.
The measurement required during the protocol is not dependent on the calculation, but the final measurement result, which is the output of the calculation, is dependent on the calculation.
Therefore, we randomly perform an $X$ gate or an identity gate at the end of each calculation.
The user accepts the result of the identity gate being executed and flips and accepts the result of the $X$ gate being executed.
In this case, the probability of getting output either ``0'' or ``1'' as the server's measurement result is $50\%$ each, and the actual output of the calculation cannot be known from the measurement result.
Next, we show how to execute the $X$ gate or the identity gate without the server's knowledge.
X-axis rotation on the Bloch ball can be implemented as follows:
\begin{equation}
    R_x(\frac{\pi}{4}) = H\cdot T \cdot H \cdot I.
\end{equation}
If this $R_x(\frac{\pi}{4})$ is applied four times, it becomes the $X$ gate, and if it is applied eight times, it becomes the identity gate.
In other words, when the number of times $R_x(\frac{\pi}{4})$ is executed in a certain gate sequence is divided by 8, the remainder of 4 or 0 changes whether it is the $X$ gate or the identity gate.
Thus, the user can hide the number of $R_x(\frac{\pi}{4})$ executions by bringing the gate sequence so long that it is not known by some servers that perform classical communication after the computation, and execute the $X$ gate or the identity gate without the server know.

Finally, we propose the main protocol that is summarized about the encryption of the circuit.
In the following, we assume that the number of entire servers is $2N$ and that $K$ servers ($N>\lfloor K/2 \rfloor$) do classical communication after the computation.
\begin{defn}[Main protocol]
The basic structure of the protocol is the same as lemma \ref{QMIP(0)}.
However, the two servers in lemma \ref{QMIP(0)} are divided into N servers each, i.e., the whole system will consist of $2N$ servers.
We label each server as $\{A_1,A_2,\cdots,A_N\}$ and $\{B_1,B_2,\cdots,B_N\}$, then $A_i$ receives quantum states from $B_{i-1}$, executes any gate, and passes quantum states to $B_i$.
We also define that the server $B_N$ sends a quantum state to server $A_1$.
This protocol encrypts the circuit in the following process instead of using the user's circuit in lemma \ref{QMIP(0)}.
\begin{description}
\item[Step 1.] The user restructures the circuit for the calculations by using the circuit like brickwork states.
\item[Step 2.] The user decomposes $V$ and $U_i$, which compose the circuit like brickwork states, into $HTHT$, $HT^{\dagger}HT^{\dagger}$, and $HIHI$.
Note that it is necessary to ensure that the number of gates that consist $V$ is constant.
\item[Step 3.] The user adds dummy gates to the circuit so that the dummy gates are run in parallel for the gate sequence of $\lfloor K/2 \rfloor$ gates for the gates consisted in step. 2.
At this step, the number of gates added is at most $3^{K/2}$ of the original number of gates, which is efficient.
\item[Step 4.] Finally, the user executes randomly the $X$ gate or the identity gate using the gate sequence consisting of $4N$ gates that are $HTHI$ or $HIHI$ just before measuring the quantum state corresponding to the output.
The $X$ gate or the identity gate that hides the output also applies to the gate sequence of the dummy gates.
The number of gates added in this step is $4N3^{K/2}n$ for every $n$ input qubits, so it is polynomial increasing and therefore efficient.
\end{description}
The user limits the gate to $\{HT,HT^{\dagger},HI\}$ that each server executes at a time for a gate sequence of $N$ $\{HT,HT^{\dagger},HI\}$.
\end{defn}
We show that the main protocol is a blind protocol even when $K$ servers of the $2N$ servers ($N>\lfloor K/2 \rfloor$) can do classical communication after the computation.
\begin{thm}
\label{main}
The main protocol is a blind protocol even when $K$ servers of the $2N$ servers ($N>\lfloor K/2 \rfloor$) can do classical communication after the computation.
\end{thm}
\begin{proof}
Since we assume that all servers do not perform classical communication during the computation, from lemma \ref{QMIP(0)} and theorem \ref{multi}, the servers cannot get information about the user's computation.

Next, we consider what happens after the server finishes the computation delegated by the user and sends the output to the user.
The assumption is that the $K$ servers can perform classical communication after the user's computation is completed.
It is shown from lemma \ref{QMIP(0)} and theorem \ref{multi} that servers that do not perform classical communication with other servers after the computation is finished do not get information that depends on the user's computation.
Then, we describe the servers that do the classical communication with other servers.
Since servers can do classical communication with each other, the protocol that users use to prevent servers from doing things differently from the user's instructions are no longer relevant to the server, and the server can directly get information about the user's circuit.

Even if the server can obtain information about the user's circuit, we show that it cannot get information that depends on the user's computation.
Since the circuit is built using the structure of the circuit like brickwork states, the server is unable to get information about the computation from the circuit structure. 
The server also gets information about the circuit's consecutive gates at most $\lfloor K/2 \rfloor$, but the circuit uses dummy gates parallel with the original gates.
Since the server cannot distinguish between the dummy gates and the original gates because it does not know the original user's calculations, the gate information that the server obtained is all the possible gate combinations it could get.
Hence, the server cannot get information that depends on the user's calculations from the gate information.
The input can be decomposed into $\ket{0}$ and gates without loss of generality; therefore, it can be hidden just like the gate.
Since the measurement results during the computation process do not depend on the user's original calculation, the server cannot obtain information that depends on the user's calculation from them.
The server also gets some of the measurement results that correspond to the output of the calculation.
However, recall that the user encrypted this output using the $X$ gate or the identity gate randomly.
The server cannot distinguish between the $X$ gate and the identity gate without knowing all of the gates in the last $4N$ gate sequences.
By the assumption, the server know only about $4\lfloor K/2 \rfloor$ gates out of $4N$, the server cannot know about the gates that encrypt the output.
Therefore, the server's output is half ``0'' and half ``1'', and the server cannot decrypt it, so the output does not depend on the user's calculations.
The above result holds that even if the server performing the classical communication is less than $K$.
The classical information available to the server does not depend on the user's original calculation.

The server's quantum state is identical to the two servers protocol.
If those servers can get a quantum state that depends on the calculation, servers can also get the quantum at the time of the two servers.
This contradicts lemma \ref{QMIP(0)}.
Therefore, the quantum state obtained by the server is independent of the user's original calculation.

Hence, the main protocol is a blind protocol even if $K$ servers of the $2N$ servers ($N>\lfloor K/2 \rfloor$) can do classical communication after the computation.
\end{proof}

\subsection{Risk Estimation}
Theorem \ref{main} is based on the premise that after the computation, only $K$ servers of the $2N$ servers ($N>\lfloor K/2 \rfloor$) perform classical communication after the computation.
In reality, we can assume that a user has made a contract with all servers not to do classical communication with each other, but some of them have done so in violation of the contract after the computation.
Assume that $t$ is the average time between one server leaking information and the next, whether consciously or unconsciously.
If a user chooses a sufficiently large $K$, the law of large numbers allows the user to estimate that the time it takes for the server to get the average user's information is $(K+1)t$.
Of course, if $(K+1)t$ is too long, the value of $t$ may become obsolete.
Furthermore, although t is considered fixed here, it is not necessarily fixed, and in practice, accurate model design for $t$ is necessary.
However, the main protocol allows the user to choose the parameter $K$, allowing him to adopt the risk of time leaking information that depends on the parameter $K$ rather than the risk of time leaking information by a single server.
Thus, users can estimate the risk that the previous study could not by using the main protocol, and they can choose the number of servers according to the risk they are willing to accept.

\section{Discussion}
In this paper, we proposed a blind protocol with multiple servers that is effective even if some servers perform classical communication after the computation.
This main protocol is an extension of the blind protocol with two servers proposed in the previous study\cite{Reichardt}.
We first increased the number of servers by splitting the role into multiple servers inside the server for the two servers used in the protocol of the previous study.
Next, we proposed a method of encrypting the circuit such that even if some circuit information is leaked, the server will not be able to determine which information is dependent on the user's original calculation.
We then proposed the main protocol, which summarized them and showed a blind protocol even in assumption.
Finally, we discussed how risk could be estimated, and the number of servers increased so that users can manage their own acceptable risk.

One of the disadvantages of our protocol is that it uses many more gates in the calculation than the protocol used in the previous study.
The amount it increases depends on how much risk the user is willing to accept.
However, the increase in the number of gates fits into the polynomial size.
Another disadvantage is that to use $N$ servers, $N$ quantum cloud servers sharing entanglements should exist in reality.

Another concern is that it cannot manage the situation where all servers perform classical communication after the computation or perform classical communication during the computation.
The former is especially significant. 
If the former can be solved, the blind protocol on a single server will be as secure as the blind protocol on a single server merely by not having servers perform classical communication with each other during computation.
It is a significant open problem to find a blind protocol that prevents servers from getting information dependent on the user's original calculation, even if it allows all servers to perform classical communication after the computation.

\section*{Acknowledgment}
We would like to thank Takayuki Miyadera for many helpful comments for the paper and advice to the protocol.
This preprint has not undergone peer review or any post-submission improvements or corrections. The Version of Record of thisarticle is published in Quantum Information Processing, and is available online at https://doi.org/10.1007/s11128-022-03430-y.

\end{document}